\documentclass[11pt,letterpaper]{article}

\usepackage{arxiv}

\usepackage[utf8]{inputenc} \usepackage[T1]{fontenc}    \usepackage{lmodern}        \usepackage{hyperref}       \usepackage{url}            \usepackage{booktabs}       \usepackage{amsfonts}       \usepackage{nicefrac}       \usepackage{microtype}      \usepackage{graphicx}

\title{Likelihood-based Inference for Exponential-Family Random Graph
Models via Linear Programming}

\author{
    Pavel N.~Krivitsky
   \\
    Department of Statistics \\
    University of New South Wales \\
  Sydney, NSW, Australia \\
  \texttt{\href{mailto:p.krivitsky@unsw.edu.au}{\nolinkurl{p.krivitsky@unsw.edu.au}}} \\
   \And
    Alina R.~Kuvelkar
   \\
    Department of Statistics \\
    Penn State University \\
  State College, PA, USA \\
  \texttt{\href{mailto:ark336@psu.edu}{\nolinkurl{ark336@psu.edu}}} \\
   \And
    David R.~Hunter
   \\
    Department of Statistics \\
    Penn State University \\
  State College, PA, USA \\
  \texttt{\href{mailto:dhunter@stat.psu.edu}{\nolinkurl{dhunter@stat.psu.edu}}} \\
  }

\usepackage{color}
\usepackage{fancyvrb}

\DefineVerbatimEnvironment{Highlighting}{Verbatim}{commandchars=\\\{\}}
\usepackage{framed}
\definecolor{shadecolor}{RGB}{248,248,248}
\newenvironment{Shaded}{\begin{snugshade}}{\end{snugshade}}

\newcommand{\CommentTok}[1]{\textcolor[rgb]{0.56,0.35,0.01}{\textit{#1}}}

\newcommand{\ControlFlowTok}[1]{\textcolor[rgb]{0.13,0.29,0.53}{\textbf{#1}}}
\newcommand{\DataTypeTok}[1]{\textcolor[rgb]{0.13,0.29,0.53}{#1}}
\newcommand{\DecValTok}[1]{\textcolor[rgb]{0.00,0.00,0.81}{#1}}

\newcommand{\FloatTok}[1]{\textcolor[rgb]{0.00,0.00,0.81}{#1}}

\newcommand{\KeywordTok}[1]{\textcolor[rgb]{0.13,0.29,0.53}{\textbf{#1}}}
\newcommand{\NormalTok}[1]{#1}
\newcommand{\OperatorTok}[1]{\textcolor[rgb]{0.81,0.36,0.00}{\textbf{#1}}}
\newcommand{\OtherTok}[1]{\textcolor[rgb]{0.56,0.35,0.01}{#1}}

\newcommand{\StringTok}[1]{\textcolor[rgb]{0.31,0.60,0.02}{#1}}

\usepackage[round,authoryear]{natbib}
\usepackage{amsmath,amssymb,tikz,enumitem,bm,booktabs,etoolbox,array,amsthm,subfig,doi,mathtools,textcomp}
\AtBeginEnvironment{CodeChunk}{\footnotesize}
\AtBeginEnvironment{verbatim}{\footnotesize}
\AtBeginEnvironment{Highlighting}{\footnotesize}

\date{}
 \DeclareMathOperator{\Normal}{Normal}
\def\Reals{\mathbb{R}}
\def\Expectation{\operatorname{\mathbb{E}}}
\def\Variance{\operatorname{\mathbb{V}\kern-0.072em ar}}
\def\SampleSpace{\mathcal{Y}}

\def\ConvexHull{\mathcal{C}}
\def\obs{\text{obs}}

\def\0V{\bm{0}}
\def\1V{\bm{1}}

\def\aV{\bm{a}}

\def\eV{\bm{e}}
\def\gV{\bm{g}}

\def\pV{\bm{p}}

\def\tV{\bm{t}}

\def\wV{\bm{w}}
\def\xV{\bm{x}}
\def\yV{\bm{y}}
\def\zV{\bm{z}}

\def\thetaV{\bm{\theta}}
\def\etaV{\bm{\eta}}
\def\xiV{\bm{\xi}}

\def\mM{{M}}
\def\tM{{T}}
\def\rM{{R}}

\def\yM{{Y}}
\def\zM{{Z}}

\def\tS{{T}}
\def\sS{{S}}

\def\T{{^\top}}
\def\inv{^{-1}}

\newtheorem{lemma}{Lemma}

\newtheorem{proposition}{Proposition}
\newtheorem{corollary}{Corollary}

\newcommand{\lpref}[1]{LP~\eqref{#1}}
\newcommand{\figref}[1]{Figure~\ref{#1}}

\newcommand{\secref}[1]{Section~\ref{#1}}

\newcommand{\abs}[1]{\lvert #1 \rvert}

\def\bse{\begin{subequations}}
\def\ese{\end{subequations}}

\begin{document}
\maketitle

\begin{abstract}
This article discusses the problem of determining whether a given point,
or set of points, lies within the convex hull of another set of points
in \(d\) dimensions. This problem arises naturally in a statistical
context when using a particular approximation to the loglikelihood
function for an exponential family model; in particular, we discuss the
application to network models here. While the convex hull question may
be solved via a simple linear program, this approach is not well known
in the statistical literature. Furthermore, this article details several
substantial improvements to the convex hull-testing algorithm currently
implemented in the widely used \texttt{ergm} package for network
modeling.
\end{abstract}


\clearpage

\hypertarget{sec:intro}{\section{Monte Carlo maximum likelihood estimation for exponential
families}\label{sec:intro}}

Suppose that we observe a complex data structure \(\yM\) from a sample
space \(\SampleSpace\), and we believe that the process that produced
\(\yM\) may be captured sufficiently by the \(d\)-dimensional vector
statistic \(\gV:\SampleSpace\mapsto \Reals^d\). A natural probability
model for such an object---one that minimises the additional assumptions
made in the sense of maximising entropy---is the exponential family
class of models \citep{Ja57i}. If \(\yM\) modelled via a discrete or a
continuous distribution, the probability mass function or density has
the form, parametrized by the \emph{canonical parameter}
\(\thetaV\in \Reals^d\) \citep[ among others]{barndorffnielsen1978}, \[
p_{\thetaV; \SampleSpace, h, \gV}(\yM) = h(\yM) 
\frac{\exp\{\thetaV\T\gV(\yM)\}}{\kappa_{\SampleSpace, h, \gV}(\thetaV)},\ \yM\in\SampleSpace.\label{eq:ef}
\] The form is then a product of a function \(h(\cdot)\) of the data
alone specifying the distribution under \(\thetaV=\0V\), a function
\(\kappa(\cdot)\) of the parameter alone, defined below, and an
exponential with the data's and parameter's interaction. The exact form
of \(\kappa(\thetaV)\) depends on the type of distribution:
\bse\label{eq:normc} \begin{align}
&\text{discrete:}&\kappa_{\SampleSpace, h, \gV}(\thetaV) &= 
\sum_{\yM'\in \SampleSpace} h(\yM) \exp\{ \thetaV\T \gV (\yM') \}\label{eq:normc-discrete}\\
&\text{continuous:}&\kappa_{\SampleSpace, h, \gV}(\thetaV) &= \int_{\SampleSpace}
h(\yM) \exp\{ \thetaV\T \gV (\yM') \} d\yM'.\label{eq:normc-continuous}\end{align} \ese (A full measure-theoretic formulation is also possible
for distributions that are neither discrete nor continuous; and a
\(\thetaV\in \Reals^q\) for \(q\le d\) may be mapped through a vector
function \(\etaV: \Reals^q\mapsto \Reals^d\).) Likelihood-based
inference for exponential-family models centers on the log-likelihood
function \begin{equation}\label{eq:Loglikelihood}
\ell(\thetaV) = \thetaV\T \gV (\yM_{\obs }) - \log \kappa_{\SampleSpace, h, \gV} (\thetaV ).
\end{equation} Likelihood calculations can be highly computationally
challenging when the sum or integral \eqref{eq:normc} is intractable
\citep[e.g.,][]{geyer1992}.

For the sake of brevity, we will omit ``\(\SampleSpace, h, \gV\)'' from
the subscript of \(p_{\thetaV; \SampleSpace, h, \gV}(\yM)\) and
\(\kappa_{\SampleSpace, h, \gV}(\thetaV)\) for the remainder of this
paper, unless they differ from these defaults.

Most commonly, exponential families are used as a tool for deriving
inferential properties of families that belong to this class, albeit
with a different parametrization. For a common example, the canonical
parameter of the \(\Normal(\mu,\sigma^2)\) distribution in its
exponential family form is \(\thetaV=[\mu\sigma^{-2},\sigma^{-2}]\),
which is far less convenient and interpretable.

However, there are domains in which an exponential family model is
specified directly through its sufficient statistics, often with the
help of the Hammersley--Clifford Theorem \citep[among others]{Be74s}.
These include the Strauss spatial point processes \citep{St75m}, the
Conway--Maxwell--Poisson model for count data \citep{ShMi05u}, and
exponential-family random graph models (ERGMs) \citep[ for original
derivations]{HoLe81e, FrSt86m} for networks and other relational data.
While our development is motivated by and focuses on ERGMs, it applies
to other scenarios involving exponential-family models with intractable
normalizing constants because the methods we discuss operate on the
sufficient statistic \(\gV(\yM)\) rather than on the original data
structure.

Consider a graph \(\yM\) on \(n\) vertices, with the vertices being
labelled \(1,\dotsc,n\). We will focus on binary undirected graphs with
no self-loops and no further constraints---a discrete distribution.
Thus, we can express
\(\SampleSpace=2^{\{\{i,j\}\in \{1,\dotsc,n\}^2: i\ne j\}}\), the power
set of the set of all distinct unordered pairs of vertex indices.

An exponential family on such a sample space is called an
exponential-family random graph model (ERGM). Substantively, elements of
\(\gV(\yM)\) then represent features of the graph---e.g., the number of
edges or other structures, or connections between exogenously defined
groups---whose prevalence is hypothesised to influence the relative
likelihoods of different graphs.

For example, an edge count statistic would, through its corresponding
parameter, control the relative probabilities of sparser and denser
graphs, and in turn the expected density of the graph, conditional on
other statistics. Some of the graph statistics, the most familiar being
the number of triangles, or ``friend of a friend is a friend,''
configurations, induce stochastic dependence among the edges. The
triangle statistic in particular is problematic for reasons reviewed by
a number of authors \citep[e.g., Section 3.1 of][]{ScKr20e} and has been
largely superseded by statistics proposed by \citet{SnPa06n} and others.

For our special case, \(\kappa(\thetaV)\) has the form
\eqref{eq:normc-discrete}, a summation over all possible graphs. For the
population of binary, undirected, vertex-labelled graphs with no
self-loops, the cardinality of \(\SampleSpace\) is \(2^{\binom{n}{2}}\),
a number exponential in the square of the number of vertices. Thus, even
for small networks, this sum is intractable. For example, for \(n=10\),
summation is required of
\(\abs{\SampleSpace} \approx \ensuremath{3.5\times 10^{13}}\)
elements---too many to compute by ``brute force.'' Under some choices of
\(\gV(\cdot)\), the summation \eqref{eq:normc-discrete} simplifies, but
for most interesting models---those involving complex dependence among
relationships in the network represented by the graph---maximization of
\(\ell(\thetaV)\) can be an enormous computational challenge.

Various authors have proposed methods for approximately maximizing the
log-likelihood function. For instance, \citet{snijders2002} introduced a
Robbins--Monro algorithm \citep{RoMo51s} based on the fact that when
\(\hat{\thetaV}\) denotes the maximum likelihood estimator (MLE), \[
\Expectation_{\hat{\thetaV}}[\gV (\yM)] = 
\gV (\yM_{\obs}).
\] Alternatively, the MCMC MLE idea of \citet{geyer1992} was adapted to
the ERGM framework by, among others, \citet{hummel2012}. This idea is
based on the fact that the log-likelihood-ratio \[
\lambda(\thetaV,\thetaV_0) \equiv \ell(\thetaV) - \ell(\thetaV_{0}) = (\thetaV-\thetaV_{0})\T \gV (\yM_{\obs}) - 
\log \Expectation_{\thetaV_0}[\exp\{(\thetaV - \thetaV_{0})\T \gV ({\SampleSpace})\} ],
\] suggesting that if we sample \(r\) networks
\(\yM_{1},\dotsc,\yM_{r}\) from the approximate distribution
\(p_{\thetaV_{0}}(\cdot)\) via MCMC, we can employ an estimator
\begin{equation}\label{eq:LoglikApprox}
\hat{\lambda}(\thetaV,\thetaV_0) \equiv (\thetaV - \thetaV_{0})\T \gV (\yM _{\obs }) - 
\log\left[\frac{1}{r} \sum_{i=1}^r \exp\{(\thetaV - \thetaV_{0})\T \gV ({\yM_i})\}\right]
\end{equation} as an approximation to \(\lambda(\thetaV, \thetaV_{0})\).

In some cases, the network may be partially unobserved, i.e.,
\(\yM_{\obs }\) is not a complete network in \(\SampleSpace\). Letting
\(\SampleSpace(\yM_{\obs})\) denote the set of all networks in
\(\SampleSpace\) that coincide with \(\yM_{\obs}\) wherever
\(\yM_{\obs}\) is observed, we may generalize \eqref{eq:Loglikelihood}
as in \citet{handcock2010} by writing \begin{equation}\label{eq:Loglik2}
\ell(\thetaV) =  \log\sum_{\yM'\in \SampleSpace(\yM_{\obs})} p_{\thetaV}(\yM') =  
\log \kappa_{\SampleSpace(\yM_{\obs})}(\thetaV ) - \log \kappa_{\SampleSpace}(\thetaV ).
\end{equation} Equation \eqref{eq:Loglik2} generalizes
\eqref{eq:Loglikelihood} because \(\SampleSpace(\yM_{\obs})\) consists
of the singleton \(\{\yM_{\obs}\}\) when the network is fully observed.

This likelihood may be approximated using the approach of
\citet{GeCa93m} and \citet{Ge94c}: Draw samples
\(\yM_{1},\dotsc,\yM_{r}\) and \(\zM_{1},\dotsc,\zM_{s}\) from
\(\SampleSpace\) and \(\SampleSpace(\yM_{\obs})\), respectively via
MCMC, such that the stationary distributions are
\(p_{\thetaV_0;\SampleSpace}(\cdot)\) and
\(p_{\thetaV_0;\SampleSpace(\yM_{\obs})}(\cdot)\), respectively. The
generalization of \eqref{eq:LoglikApprox} is then
\begin{multline}\label{eq:LoglikApprox2}
\hat{\lambda}(\thetaV,\thetaV_{0}) \equiv 
\log\left[\frac{1}{s} \sum_{i=1}^s \exp\{(\thetaV - \thetaV_{0})\T \gV ({\zM_i})\}\right]\\-
\log\left[\frac{1}{r} \sum_{i=1}^r \exp\{(\thetaV - \thetaV_{0})\T \gV ({\yM_i})\}\right].
\end{multline}

It seems that \eqref{eq:LoglikApprox2} allows us to find an approximate
maximum likelihood estimator by simply maximizing over \(\thetaV\). Yet
there are problems with this approach, and in this paper we focus on one
problem in particular: It is not always the case that
\eqref{eq:LoglikApprox2} has a maximizer, nor even an upper bound.

To characterize precisely when the approximation of
\eqref{eq:LoglikApprox2} has a maximizer, we first define a term that
will be important throughout this article: The \emph{convex hull} of any
set of points in \(d\)-dimensional Euclidean space is the smallest
convex set containing that set, i.e., the intersection of all convex
sets containing that set. According to exponential family theory
\citep[e.g., Theorem 9.13 of][]{barndorffnielsen1978},
\(\hat{\lambda}(\thetaV,\thetaV_0)\) in \eqref{eq:LoglikApprox} has a
maximizer if and only if \(\gV(\yM_\obs)\) is contained in the
\emph{interior} of the convex hull of
\(\{\gV(\yM_1), \dotsc, \gV(\yM_r)\}\). Indeed, it is straightforward to
show that \(\hat{\lambda}(\thetaV,\thetaV_0)\) in the more general
expression \eqref{eq:LoglikApprox2} does not have a maximum if
\emph{any} of \(\{\gV(\zM_1), \dotsc, \gV(\zM_s)\}\) lies outside the
convex hull of \(\{\gV(\yM_1), \dotsc, \gV(\yM_r)\}\). We prove this in
\secref{sec:LinearProgram}.

The question of determining when a point lies within the convex hull of
another set of points is thus relevant when using approximations
\eqref{eq:LoglikApprox} and \eqref{eq:LoglikApprox2}. The remainder of
this article shows how to determine whether a given point or set of
points lies inside a given convex hull using linear programming and
explains how the \texttt{ergm} package for R exploits this fact, then
develops improvements to the algorithm and extending them to the case
where a network is only partially observed.

\hypertarget{sec:LinearProgram}{\section{Convex hull testing as a linear
program}\label{sec:LinearProgram}}

We introduce the terms \emph{target set} and \emph{test set} to refer to
the set \(\tS=\{\gV (\yM_{1}),\dotsc,\gV (\yM_{r})\}\) and the set
\(\sS=\{\gV (\zM_{1}),\dotsc,\gV (\zM_{s})\}\), respectively. To
reiterate, the convex hull of any set of \(d\)-dimensional points is the
smallest convex set containing that set, and the convex hull is always
closed in \(\Reals^d\). The interior of this convex hull will play a
special role here, so we let \({\ConvexHull(\tS)}\) denote the interior
of the convex hull of \(\tS\). If the points in \(\tS\) satisfy a linear
constraint, \({\ConvexHull(\tS)}\) is empty since in this case the
convex hull lies in an affine subspace of dimension smaller than \(d\).
We assume here that \({\ConvexHull(\tS)}\) is nonempty, which in turn
means that the convex hull is the closure of \({\ConvexHull(\tS)}\),
which we denote by \(\overline{\ConvexHull(\tS)}\).

As mentioned near the end of \secref{sec:intro}, the approximation in
\eqref{eq:LoglikApprox2} fails to admit a maximizer whenever
\(\sS\not\subset\overline{\ConvexHull(\tS)}\), which may be proved quite
simply:

\begin{proposition}
If \(\sS\not\subset\overline{\ConvexHull(\tS)}\), then
\(\sup_{\thetaV} \hat{\lambda}(\thetaV,\thetaV_{0}) = \infty\).

\end{proposition}

\begin{proof}
Suppose there exists an element of \(\sS\), say \(\gV (\zM_1)\), in the
open set \(\Reals^d \setminus \overline {\ConvexHull(\tS)}\). By the
convexity of \(\overline {\ConvexHull(\tS)}\), this means there exists a
hyperplane \({\cal H}\)---that is, an affine \((d-1)\)-dimensional
subspace---containing \(\gV (\zM_1)\) and such that
\({\cal H} \cap \overline {\ConvexHull(\tS)} = \emptyset\). In other
words, there exist a scalar \(z _ 0\) and a \(d\)-vector \(\zV\) such
that \(z _ 0 + \zV\T\gV(\zM _ 1)=0\) and
\(z _ 0 + \zV\T\gV(\yM _ i) < 0\) for \(i = 1, \dotsc, r\).

If we now let \(\thetaV = \thetaV _ 0 + \alpha \zV\) for a real number
\(\alpha\), \eqref{eq:LoglikApprox2} gives \begin{align*}
\exp\{\hat{\lambda}(\thetaV,\thetaV_0)\}
&=\frac{r}{s} \times \frac{\exp\{\alpha\zV\T \gV ({\zM_1})\} + \sum_{i=2}^s 
\exp\{\alpha\zV\T \gV ({\zM_i})\} } {\sum_{i=1}^r \exp\{\alpha\zV \T \gV ({\yM_i})\}}\\
&\ge\frac{r}{s} \times \frac{1} {\sum_{i=1}^r \exp[\alpha\zV \T \{\gV({\yM_i})-\gV 
({\zM_1})\}]}.
\end{align*} Since \(z _ 0 + \zV\T\gV(\zM _ 1)=0\) and
\(z _ 0 + \zV\T\gV(\yM _ i) < 0\) imply that
\(\zV \T \{\gV({\yM_i})-\gV({\zM_1}) \} < 0\) for \(i = 1, \dotsc, r\),
every term in the denominator goes to zero as \(\alpha\to\infty\).

\end{proof}

However, \(\sS\subset\overline{\ConvexHull(\tS)}\) is merely a
necessary, not sufficient, condition for \eqref{eq:LoglikApprox2} to
have a unique maximizer. Another necessary condition is that
\(\frac{\partial^2\hat{\lambda}(\thetaV,\thetaV_0)}{\partial\thetaV}\)
must be negative-definite for all \(\thetaV\) and \(\thetaV_0\), but it
is not sufficient either, since a direction of recession \citep[for
example]{Ge09l} may exist. We are thus not aware of a necessary and
sufficient condition when \(\sS\) contains more than one point. As we
see below, for our purposes, the conditions that
\(\sS\subset\ConvexHull(\tS)\) and that
\(\Variance(\tS)-\Variance(\sS)\) is positive-definite suffice.

We now demonstrate that linear programming can provide a method of
testing definitively whether \(\sS\subset\overline{\ConvexHull(\tS)}\).
As a practical matter, we can apply this method to a slightly perturbed
version of \(\sS\) in which each test point is expanded away from the
centroid of \(\tS\) by a small amount; if the perturbed \(\sS\) is
contained in \(\overline{\ConvexHull(\tS)}\), then
\(\sS\subset{\ConvexHull(\tS)}\). This approach has the added benefit of
ensuring not only that \(\sS\subset{\ConvexHull(\tS)}\) but that each
point in \(\sS\) is bounded away from the boundary of
\({\ConvexHull(\tS)}\) by a small amount, which can prevent some
computational challenges in using the approximation in
\eqref{eq:LoglikApprox2}. Yet we will also show, in
\secref{sec:reformulation}, that we may transform the linear program to
test directly whether \(\sS\subset\ConvexHull(\tS)\).

In the particular case in which \(\yM_{\obs}\) is a full network, i.e.,
there are no missing data, \(\sS\) is the singleton
\(\{ \gV(\yM _{\obs})\}\). This is the case considered by
\citet{hummel2012}, who propose an approximation method that checks
\(\sS\subset{\ConvexHull(\tS)}\) and then replaces \(\gV(\yM_{\obs})\)
by some other point contained in \({\ConvexHull(\tS)}\) whenever
\(\sS\not\subset{\ConvexHull(\tS)}\). In particular, the method defines
a ``pseudo-observation'', \(\hat{\xiV}\), as the convex combination
\(\gamma \gV(\yM_{\obs}) + (1 - \gamma) \overline{\tV}\), where
\(\overline{\tV}\) is the centroid of \(\tM\) and \(\gamma \in (0, 1]\).
A previous implementation of the algorithm in the \texttt{ergm} package
\citep{hunter2008} chooses the largest \(\gamma\) value such that
\(\hat{\xiV} \in {\ConvexHull(\tS)}\) via a grid search on \((0,1]\),
then maximizes the resulting approximation to yield a new parameter
vector \(\tilde{\thetaV}\), which then replaces \(\thetaV_0\) in
defining the MCMC stationary distribution, and the process repeats until
\(\gV(\yM_{\obs}) \in {\ConvexHull(\tS)}\) for two consecutive
iterations. The current \texttt{ergm}-package implementation of the
algorithm uses a bisection search method that incorporates a prior
belief on the value of \(\gamma\) using the same condition that
\(\hat{\xiV} \in {\ConvexHull(\tS)}\).

One downside of these algorithms is that they are all
``trial-and-error'' algorithms requiring multiple checks of the
condition \(\hat{\xiV} \in {\ConvexHull(\tS)}\). In this article, we
show how to eliminate this ``trial-and-error'' approach.

We first frame the check of whether
\(\pV\in\overline{\ConvexHull(\tS)}\), for an arbitrary column vector
\(\pV\in\Reals^d\), as a linear program. Let \(\mM\) be the
\((r\times d)\)-dimensional matrix whose rows \(\mM_{1},\dotsc,\mM_{r}\)
are the points in the target set \(\tS\); furthermore, let
\(\ConvexHull(\mM)\) denote \(\ConvexHull(\tS)\). Because
\(\overline{\ConvexHull(\mM)}\) is convex in \(\Reals^d\), for any
\(\pV \notin \overline{\ConvexHull(\mM)}\), we may find an affine
\((d-1)\)-dimensional subspace, which we call a \emph{separating
hyperplane}, that separates the points in \(\tS\) from the point \(\pV\)
in the sense that \(\pV\) lies in one (open) half-space defined by the
hyperplane and the points in \(\tS\)---i.e., the rows of \(\mM\)---all
lie in the other (closed) half-space. Mathematically, this separating
hyperplane is determined by the affine subspace
\(\{\xV\in\Reals^{d}:z_0 + \xV\T\zV =0\}\) for some scalar \(z_0\) and
\(d\)-vector \(\zV\). Thus, \(\zV\in\Reals^d\) and \(z_0\in\Reals\)
determine a separating hyperplane whenever \(z_0 + \pV\T \zV < 0\) while
\(z_0 + \mM_i \zV \ge 0\) for each row \(\mM_i\) of \(\mM\).

In other words, a hyperplane separating \(\pV\) from
\(\overline{\ConvexHull(\mM)}\) exists whenever the minimum value of
\(z_0 +\pV \T\zV\) is strictly negative, where the minimum is taken over
all \(z_0\in\Reals\) and \(\zV =(z_1, \dotsc, z_d)\T \in \Reals^d\) such
that \(\mM_i\zV \ge -z_0\) for \(i=1,\dotsc,n\). Notationally, we will
use inequality symbols to compare vectors componentwise, so for example
instead of ``\(\mM_i\zV \ge -z_0\) for \(i=1,\dotsc,n\)'', we may write
``\(\mM\zV \ge -z_0\1V\)''. The existence of a separating hyperplane can
thus be determined using the following linear program:
\begin{equation}\label{LPoriginal}
\begin{aligned}
  & \text{minimize: } & & z_0 + \pV\T\zV \\
   & \text{subject to: }& \quad &  \begin{aligned}[t] \mM\zV & \ge -z_0 \1V  \\
                  \zV  & \ge -\1V \\
                  \zV & \le \1V.
                  \end{aligned}
                  \end{aligned}
\end{equation} If the minimum value of the objective function
\(z_0+\pV \T\zV\) is exactly 0---it can never be strictly positive
because \((z_0,\zV)=\0V_{d+1}\) determines a feasible point---then
\(\pV\) is in \(\overline{\ConvexHull(\mM)}\), the closure of
\(\ConvexHull(\mM)\). If the minimum value is strictly negative, then a
separating hyperplane not containing \(\pV\) exists and thus
\(\pV \not\in\overline{\ConvexHull(\mM)}\). The bounds
\(-\1V \le \zV \le \1V\), which we call \emph{box constraints}, ensure
that the linear program has a finite minimizer when
\(\pV \not\in\overline{\ConvexHull(\mM)}\): If \(\zV\) is unconstrained
and there exists a feasible \(z_0\), \(\zV\) with
\(z_0 + \pV\T\zV < 0\), then \((cz_0) + (c\pV)\T\zV < z_0 + \pV\T\zV\)
for any \(c>1\) and no finite minimizer exists. The box constraints are
implemented in the current \texttt{ergm} package version of the convex
hull check; however, in \secref{sec:reformulation} we show how to
eliminate them entirely, leading to a more streamlined linear program.

The current implementation of the algorithm in the \texttt{ergm} package
uses the interface package \texttt{lpSolveAPI} to solve the linear
program of \ref{LPoriginal}, which we refer to as ``\lpref{LPoriginal}''
henceforth.

In case the test set \(\sS\) contains more than one point, we may test
each point individually to determine whether each point in \(\sS\) is
contained in \(\overline{\ConvexHull(\mM)}\).

\hypertarget{reformulating-the-single-test-point-algorithm}{\section{\texorpdfstring{Reformulating the single-test-point algorithm
\label{sec:reformulation}}{Reformulating the single-test-point algorithm }}\label{reformulating-the-single-test-point-algorithm}}

The trial-and-error algorithm of \citet{hummel2012}, which applies only
to the case where the test set \(\sS\) consists of the single point
\(\pV\), never exploits the fact that the separating hyperplane, if it
exists, is the minimizer of a function. We propose to improve this
algorithm by reformulating the linear program and then using the
minimizers found. First, we establish a few basic facts that help us
simplify the problem.

Since translating \(\pV\) and all points in \(\tS\) by the same constant
\(d\)-vector does not change whether
\(\pV\in\overline{\ConvexHull(\mM)}\), we assume throughout this section
that all points are translated so that the origin is a point in
\(\ConvexHull(M)\) that we refer to as the ``centroid'' of \(\mM\). In
practice, we often take the centroid of \(\mM\) to be the sample mean of
the points in the target set. Yet none of the results of this section
rely on any particular definition of ``centroid''; here, we assume only
that the centroid \(\0V\in\Reals^d\) lies in the interior of the convex
hull of the target set.

Since \(\pV\) may be assumed not to coincide with the centroid, else the
question of whether \(\pV\in\ConvexHull(M)\) is answered immediately, at
least one of the coordinates of \(\pV\) is nonzero. We assume without
loss of generality that \(p_1\ne0\); otherwise, we may simply permute
the coordinates of all the points without changing whether
\(\pV\in\overline{\ConvexHull(\mM)}\). Below, we define a simple
invertible linear transformation that maps \(\pV\) to the unit vector
\(\eV_1=(1, 0, \dotsc, 0)\T\). Applying this transformation allows us to
simplify the original linear program and thereby clarifies our
exposition while facilitating establishing results; then, once we have
proved our results, we will apply the inverse transformation to restore
the original coordinates, achieving a simplification of the linear
program in the process.

One important fact about any invertible linear transformation, when
applied to \(\pV\) and \(\mM\), is that it does not change whether or
not \(\pV \in\overline{\ConvexHull(\mM)}\), as we now prove.

\begin{lemma}
If \(\rM \in \Reals^{d\times d}\) has full rank, then
\(\pV \in \overline{\ConvexHull(\mM)}\) if and only if
\(\rM\pV \in \overline{\ConvexHull(\mM\rM\T)}\).

\end{lemma}

\begin{proof}
Any point in the convex hull of a set may be written as a convex
combination of points in that set. Suppose
\(\pV \in \overline{\ConvexHull(\mM)}\). Then there exists
\(\aV \in \Reals^r\) such that \(\pV\T = \aV\T\mM\) and
\(\sum_{i=1}^r a_i = 1\). Since \((\rM\pV)\T =\aV\T\mM\rM\T\), we know
that \(\rM\pV\in \overline{\ConvexHull(\mM\rM\T)}\). Conversely, if
\(\rM\pV \in \overline{\ConvexHull(\mM\rM\T)}\), the same reasoning
applied to the full-rank transformation \(\rM\inv\) shows that
\(\pV \in \overline{\ConvexHull(\mM)}\).

\end{proof}

Consider the full-rank linear transformation defined by
\begin{equation}\label{R}
\rM = \begin{bmatrix}
\frac{1}{p_1} & 0 & \cdots & 0 \\
\frac{-p_2}{p_1} & 1 & \cdots & 0 \\
\vdots & \vdots & \ddots & 0 \\
\frac{-p_d}{p_1} & 0 & \cdots & 1 \\
\end{bmatrix} =
\begin{bmatrix}
p_1 & 0 & \cdots & 0 \\
p_2 & 1 & \cdots & 0 \\
\vdots & \vdots & \ddots & 0 \\
p_d & 0 & \cdots & 1 \\
\end{bmatrix}^{-1}
\end{equation} that maps \(\pV\) to the standard basis vector
\(\eV_1 \in \Reals^d\). After applying this linear transformation, we
want to know whether a hyperplane exists that separates the point
\(\eV_1\) from the rows of \(\mM \rM\T\).

Such a hyperplane, if it exists, may be written as
\(\{\xV\in\Reals^d: z_0 + \xV\T\zV = 0\}\) for some \(z_0\in\Reals\) and
\(\zV\in\Reals^d\). Since the origin is an interior point of
\(\ConvexHull(\mM \rM\T)\), no separating hyperplane may pass through
the origin. Therefore, we may limit our search to those cases where
\(z_0\ne0\), which means that we may divide by \(z_0\) the equation
defining the hyperplane. Stated differently, we may take \(z_0=1\) and
rewrite our hyperplane as \(\{\xV\in\Reals^d: 1 + \xV\T\zV = 0\}\)
without loss of generality.

Summarizing the arguments above, we now seek a \(d\)-vector \(\zV\)
satisfying \(1+\eV_1\T\zV<0\) and \(1 + (\mM \rM\T)_i \zV \ge0\) for
\(i=1,\dotsc,r\). Furthermore, fixing \(z_0=1\) means that no box
constraints are necessary in the reformulation of the linear program
designed to search for a separating hyperplane. Therefore, our new
linear program, after transformation by \(\rM\), becomes
\begin{equation}
\label{LPalternative}
\begin{aligned}
&\text{minimize: } & z_1 & \\
&\text{subject to: }& [\mM\rM\T]\zV &\ge -\1V. \\
\end{aligned}
\end{equation} The zero vector is a feasible point in both
\lpref{LPoriginal} and \lpref{LPalternative}. However, unlike
\lpref{LPoriginal}, zero cannot be the optimum of \lpref{LPalternative};
that is, the solution of \lpref{LPalternative} always defines a
hyperplane, whether or not it is a separating hyperplane. We now prove
this fact and show how to determine whether
\(\eV_1\in\overline{\ConvexHull(\mM\rM\T)}\):

\begin{proposition}
Let \(\zV^*\) denote a minimizer of \lpref{LPalternative}. Then
\(z_1^*<0\), and \(\pV\in\overline{\ConvexHull(\mM\rM\T)}\) if and only
if \(-1\le z_1^*\). Furthermore, the ray with endpoint at the origin and
passing through \(\eV_1\) intersects the boundary of
\(\overline{\ConvexHull(\mM\rM\T)}\) at \(-(1/z_1^*)\eV_1\).

\end{proposition}

\begin{proof}
Since \([\mM\rM\T]\0V=\0V>-\1V\), there exists some \(\epsilon > 0\)
such that \(-\epsilon \eV_1\) satisfies the constraints while reducing
the objective function \(z_1\) below zero, so \(z^*_1\) must be strictly
negative.

Now let \({\cal H} = \{\xV: 1+ \xV\T \zV^* = 0 \}\). Since \(\zV^*\) is
a feasible point, \(\overline{\ConvexHull(\mM\rM\T)}\) lies entirely in
the closed positive half-plane defined by \({\cal H}\). Thus, if
\(z^*_1<-1\), then \(\eV_1\) lies in the open negative half-plane
defined by \({\cal H}\) and so \({\cal H}\) is a separating hyperplane.
Conversely, if a separating hyperplane \(\{\xV: w_0+ \xV\T \wV = 0 \}\)
exists, then \(w_0\) must be strictly positive since \(\0V\) must be
separated from \(\eV_1\) and thus, dividing by \(w_0\), we must have
\((\mM\rM\T) (\wV/w_0) \ge -\1V\) and \(w_1/w_0<-1\).\\
Since \(z^*\) is a minimizer, this implies \(z_1^*\le w_1/w_0 < -1\).

The point \(\aV=-(1/z_1^*)\eV_1\) lies on the positive \(x_1\)-axis,
that is, on the ray with endpoint at the origin that passes through
\(\eV_1\). If \(\aV\) lies in the open set \({\ConvexHull(\mM\rM\T)}\),
then there exists \(\epsilon>0\) such that \((1+\epsilon)\aV\) also lies
in \({\ConvexHull(\mM\rM\T)}\); but this is impossible since
\((1+\epsilon)\aV\T \zV^* = -(1+\epsilon)<-1\) and thus
\((1+\epsilon)\aV\) violates the constraint that must be satisfied by
every point in \({\ConvexHull(\mM\rM\T)}\) due to convexity. On the
other hand, \(\aV\) must lie in \(\overline{\ConvexHull(\mM\rM\T)}\)
since otherwise there exists \(\epsilon>0\) such that
\((1-\epsilon)\aV\not\in \overline{\ConvexHull(\mM\rM\T)}\), which in
turn means that there exists a hyperplane separating
\(\overline{\ConvexHull(\mM\rM\T)}\) from \((1-\epsilon)\aV\) which must
therefore intersect the positive \(x_1\)-axis at a point between the
origin and \((1-\epsilon)\aV\), contradicting the fact that \(z_1^*\) is
the smallest possible such point of intersection. We conclude that
\(\aV=-(1/z_1^*)\eV_1\) must lie on the boundary of
\(\overline{\ConvexHull(\mM\rM\T)}\).

\end{proof}

Applying the full-rank transformation \(\rM^{-1}\), in order to
transform back to the original coordinates, establishes the following
corollary:

\begin{corollary}
\label{cor1}Let \(\zV^*\) denote a minimizer of \lpref{LPalternative}.
Then the ray with endpoint at the origin and passing through \(\pV\)
intersects the boundary of \(\overline{\ConvexHull(\mM)}\) at
\(-(1/z_1^*)\pV\). In particular, \(\pV\) is in \(\ConvexHull(\mM)\) if
and only if \(z_1^*>-1\).

\end{corollary}

This result gives a way to determine when \(\pV\) lies in
\(\ConvexHull(\mM)\), not merely \(\overline{\ConvexHull(\mM)}\), which
is an improvement over \lpref{LPoriginal}. It also suggests that we may
reconsider \lpref{LPalternative} after transforming by \(\rM\inv\) back
to the original coordinates. Indeed, since \(\zV\) in
\lpref{LPalternative} may take any value in \(\Reals^d\) and \(\rM\) has
full rank, there is no loss of generality in replacing \(\zV\) by
\((\rM\T)\inv\zV\), which allows us to rewrite \lpref{LPalternative} as
\begin{equation}\label{LPretransformed}
\begin{aligned}
&\text{minimize: } & \pV\T\zV & \\
&\text{subject to: }& \mM\zV &\ge -\1V. \\
\end{aligned}
\end{equation} Taking \(\zV^*\) and \(\zV^{**}\) as minimizers of
\lpref{LPalternative} and \lpref{LPretransformed}, respectively, we must
have \(z_1^* = \pV\T\zV^{**}\) since the minimum objective function
value must be the same for the two equivalent linear programs. We
conclude by Corollary \ref{cor1} that \(\pV\) is in \(\ConvexHull(\mM)\)
if and only if \(\pV\T\zV^{**}>-1\) and, furthermore, the point
\(-\pV/(\pV\T\zV^{**})\) is the point in \(\overline{\ConvexHull(\mM)}\)
closest to the origin in the direction of \(\pV\).

\hypertarget{duality}{\section{\texorpdfstring{Duality
\label{sec:duality}}{Duality }}\label{duality}}

This section is a digression in the sense that it does not lead to
improvements in the performance of the algorithms we discuss. Yet it
shows how the convex hull problem leads to a particularly simple
development of the well-known linear programming idea called duality. In
Sections \ref{sec:LinearProgram} and \ref{sec:reformulation}, we derived
linear programs to determine whether a test point \(\pV\) is in
\(\ConvexHull(\mM)\) based on the concept of a supporting hyperplane.
Here, we take a different approach, resulting in a new linear program
that is known as the dual of the original.

By definition of convexity, any point in \(\overline{\ConvexHull(\mM)}\)
is a convex combination of the rows of \(\mM\). That is,
\(\pV\in\overline{\ConvexHull(\mM)}\) if and only if there exists a
nonnegative vector \(\yV\in\Reals^n\) with \(\1V\T\yV\le 1\) and
\(\pV = \mM\T\yV\). We only require \(\1V\T\yV\le1\), not
\(\1V\T\yV=1\), because the point \(\0V\) is in
\(\overline{\ConvexHull(\mM)}\), so a convex combination of points of
\(\mM\) may place positive weight on \(\0V\). Rewriting
\(\1V\T\yV\le 1\) as \(-\1V\T\yV\ge -1\), we may therefore formulate a
linear program to test \(\pV\in\overline{\ConvexHull(\mM)}\) as follows:
\begin{equation}\label{LPdual}
\begin{aligned}
&\text{maximize: } & & -\1V\T\yV & \\
& \text{subject to: }& \quad &  \begin{aligned}[t] \mM\T\yV & = \pV && \\
                  \yV & \ge \0V. &&
                  \end{aligned}
                  \end{aligned}
\end{equation} If the maximum objective function value is found to be
strictly less than \(-1\), there is no convex combination of the rows of
\(\mM\) that yields \(\pV\) and we conclude that
\(\pV\not\in\overline{\ConvexHull(\mM)}\).

According to standard linear program theory \citep[Section
5.8]{vanderbei1997}, \lpref{LPretransformed} is precisely what is known
as the \emph{dual} of \lpref{LPdual}, and vice versa. We may develop
this theory as follows: If \(\zV\in\Reals^d\) satisfies the \(n\)
constraints in \lpref{LPretransformed}, we may multiply each inequality
by a nonnegative real number and add all of the inequalities to obtain a
valid inequality. Therefore, if \(\yV\in\Reals^d\) is a nonnegative
vector, we conclude that \(\yV\T\mM\zV \ge -\yV\T\1V\). If the
constraints in \lpref{LPdual} are completely satisfied, we may replace
\(\yV\T\mM\) by \(\pV\T\), then minimize with respect to \(\zV\) and
maximize with respect to \(\yV\) to obtain
\begin{equation}\label{WeakDuality}
\min_{\zV}\pV\T \zV \ge \max_{\yV}-\1V\T\yV.
\end{equation} Inequality \eqref{WeakDuality} is sometimes called the
\emph{weak duality theorem}. The \emph{strong duality theorem}, which we
do not prove here, says that equality holds in this case \citep[Section
5.4]{vanderbei1997}. In our convex hull problem, strong duality means
among other things that \(\pV\) is on the boundary of
\(\ConvexHull(\mM)\) if and only if the maximum objective function value
in \lpref{LPdual} equals \(-1\), since we proved this fact for the dual
linear program in Section \ref{sec:reformulation}.

\hypertarget{applications-and-benchmarks}{\section{Applications and
Benchmarks}\label{applications-and-benchmarks}}

To illustrate the results of \secref{sec:reformulation}, we first
consider problems in 2 dimensions since they are easy to visualize. Such
low-dimensional examples are also helpful since they sometimes lend
themselves to closed-form solutions of linear programs.

For the following benchmarks and examples, the following relevant
\texttt{R} \citep{R} package versions were used: \texttt{ergm} 4.1.6804,
\texttt{Rglpk} 0.6.4, and \texttt{lpSolveAPI} 5.5.2.0.17.7.

\hypertarget{a-three-point-example}{\subsection{A Three-point Example}\label{a-three-point-example}}

Let us consider a particularly simple example in which the target set
consists of the 2-vectors \((-1,0)\T\), \((a,1)\T\), and \((b,-1)\T\)
for \(a>0\) and \(b>0\). This choice guarantees that the convex hull is
a triangle containing the origin, so we shall consider the origin to be
the centroid in this example, regardless of the actual value of the mean
of the three points.

When the test point \(\pV\) equals \((1,0)\T\), we are in a situation
that could arise after transforming by \(R\). This leads to
\lpref{LPalternative}, which may be solved in closed form because in
this case the constraints become \[
\max\left\{-\frac1a - \frac{z_2}{a}, -\frac1b + \frac{z_2}{b}\right\} \le z_1 \le 1. 
\] Since the lines \(x_2=-1/a - x_1/a\) and \(x_2=-1/b + x_1/b\) have
one positive slope and one negative slope, we minimize the maximum at
the point of intersection, i.e., when \(-1/a - z_2/a = -1/b + z_2/b\),
which implies \(z_2^* = (a-b)/(a+b)\), which in turn implies
\(z_1^* = -2/(a+b)\). By simple examination, we know that \((1,0)\T\) is
interior to the convex hull exactly when the line through \((a,1)\T\)
and \((b,-1)\T\) intersects the \(x_1\)-axis at a value larger than 1.
Thus, \(\pV\in\ConvexHull(\mM)\) if and only if \((a+b)/2 > 1\), which
is equivalent to \(z_1^*= -2/(a+b) > -1\), thus verifying the result of
Corollary \ref{cor1} for this simple example.

\begin{figure}

{\centering \subfloat[\lpref{LPalternative} with test point on the $x_1$-axis: a simple closed-form solution exists.\label{fig:2dExample-1}]{\includegraphics[width=.45\textwidth]{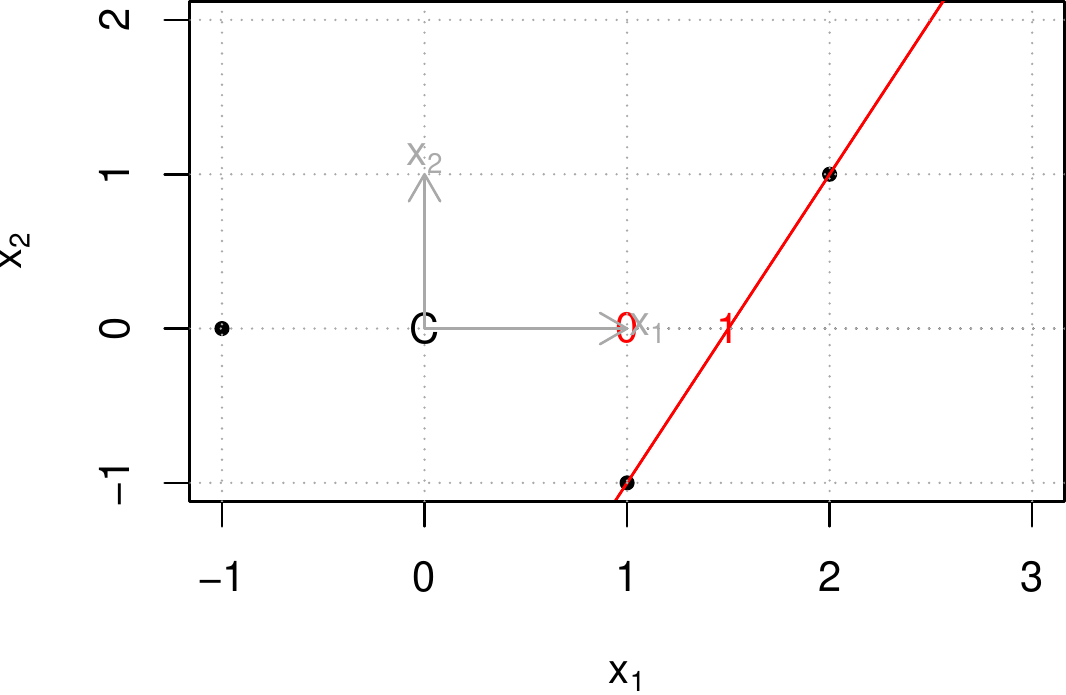} }\hspace{0.05\textwidth}\subfloat[\lpref{LPoriginal} on original scale with box constraints: a suboptimal first solution results.\label{fig:2dExample-2}]{\includegraphics[width=.45\textwidth]{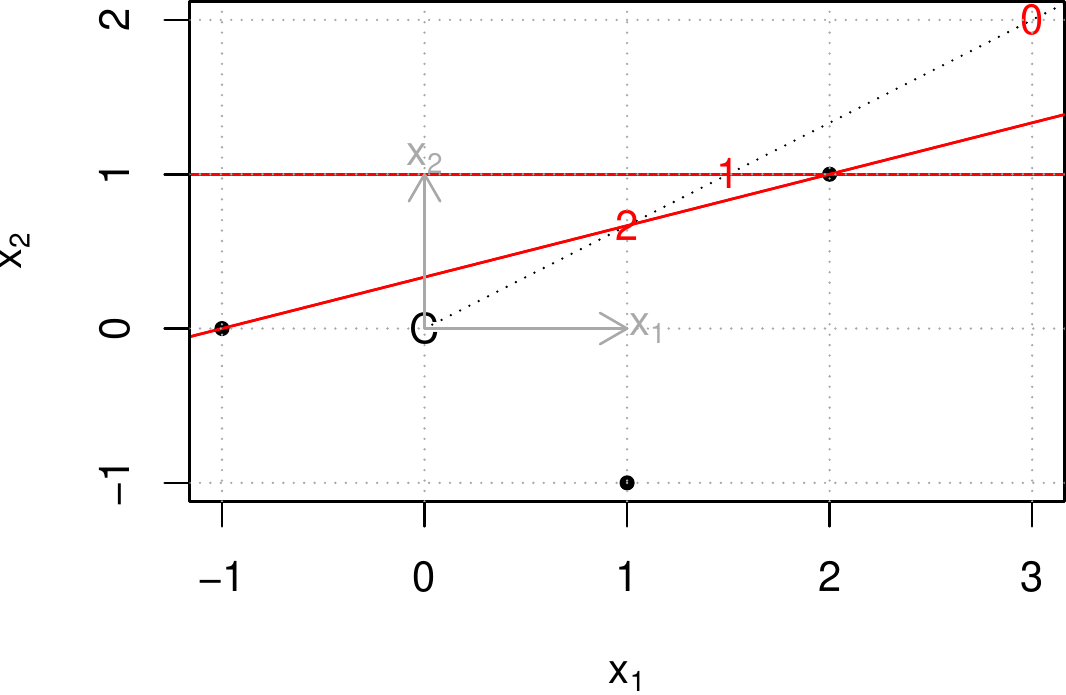} }\hspace{0.05\textwidth}\subfloat[\lpref{LPalternative} on transformed scale and without box constraints: intersection point is found immediately.\label{fig:2dExample-3}]{\includegraphics[width=.45\textwidth]{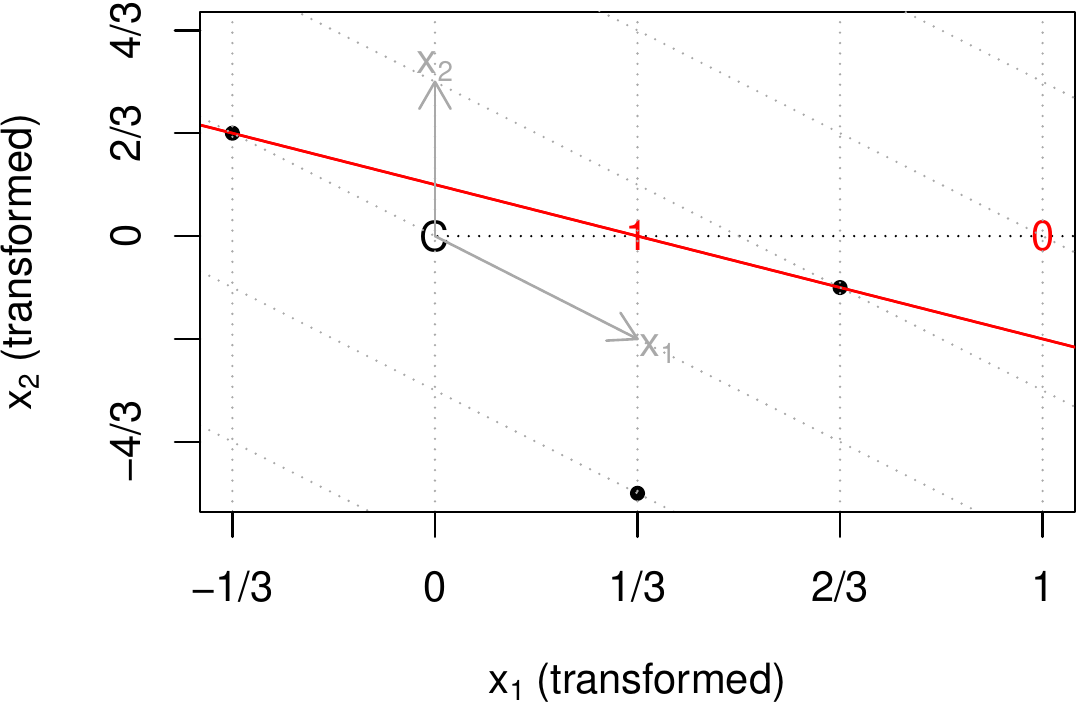} }

}

\caption{A toy example with three points. The centroid is marked with \textsf{C}; arrows indicate standard basis vectors in each dimension on the original scale, target set points are dots, supporting lines are colored red, the initial test point is \textsf{0}, and its successive iterations are labeled with integers.}\label{fig:2dExample}
\end{figure}

To make the example more concrete, let us take \(a=2\) and \(b=1\).
Then, we see in \figref{fig:2dExample-1} that \(\pV=\eV _ 1\) (labeled
with \textsf{0}) is in the convex hull, and, furthermore, that the
\(x _ 1\) coordinate of the point of intersection is \((2+1)/2=3/2\),
and the separating line's location does not depend where, along the ray
from the origin, \(\pV\) is. Now, instead, let \(\pV\) equal
\((3,2)\T\). This is the situation depicted in \figref{fig:2dExample-2}.
The original \lpref{LPoriginal}, which includes box constraints, finds
one separating line---in two dimensions, a hyperplane is a line---but
the intersection of this line and the segment connecting \(\0V\) to
\(\pV\), labeled with the digit \textsf{1}, is not optimal. A second
application of \lpref{LPoriginal} finds the optimal point labeled
\textsf{2}. In this case, point \textsf{2} lies on the line
\(\{\xV: 1+x_1+3x_2=0\}\), and it is instructive that the box
constraints alone prevent \lpref{LPoriginal} from finding this solution
when \(\pV=(3,2)\T\) but not when \(\pV=(3/2, 1)\T\), where the latter
point is the intermediate point labeled \textsf{1} in
\figref{fig:2dExample-2}.

Lastly, we illustrate the linear transformation of the previous problem
by \(\rM\). \figref{fig:2dExample-3} shows the problem transformed using
\eqref{R}, with the dotted lines and arrows showing the original
coordinate system, then \lpref{LPalternative} is applied to find the
intersection point in one iteration.

\hypertarget{sec:benchmark}{\subsection{\texorpdfstring{Benchmarking \texttt{Rglpk} against
\texttt{lpSolveAPI}}{Benchmarking Rglpk against lpSolveAPI}}\label{sec:benchmark}}

Here, we define two simple functions that each use
\lpref{LPretransformed} to search for the point on the ray from \(\0V\)
through \(\pV\) that intersects the boundary of
\(\overline{\ConvexHull(\mM)}\). Each of these functions exploits an
existing R package that wraps open-source code solving linear programs:
respectively, \texttt{Rglpk} \citep{Rglpk} wraps the GNU Linear
Programming Kit (GLPK) \citep{GLPK} and \texttt{lpSolveAPI}
\citep{lpSolveAPI} wraps the \texttt{lp\_solve} library \citep{lpsolve}.
Here is a function that uses \texttt{lpSolveAPI}:

\begin{Shaded}
\begin{Highlighting}[]
\KeywordTok{library}\NormalTok{(lpSolveAPI)}
\NormalTok{LPmod1 \textless{}{-}}\StringTok{ }\ControlFlowTok{function}\NormalTok{(M, p) \{ }\CommentTok{\# it is assumed the centroid is the zero vector}
\NormalTok{  lp \textless{}{-}}\StringTok{ }\KeywordTok{make.lp}\NormalTok{(n \textless{}{-}}\StringTok{ }\KeywordTok{nrow}\NormalTok{(M), d \textless{}{-}}\StringTok{ }\KeywordTok{ncol}\NormalTok{(M))}
  \ControlFlowTok{for}\NormalTok{ (k }\ControlFlowTok{in} \KeywordTok{seq\_len}\NormalTok{(d)) }\KeywordTok{set.column}\NormalTok{(lp, k, M[, k])}
  \KeywordTok{set.constr.type}\NormalTok{(lp, }\KeywordTok{rep}\NormalTok{(}\StringTok{"\textgreater{}="}\NormalTok{, n))}
  \KeywordTok{set.rhs}\NormalTok{(lp, }\KeywordTok{rep}\NormalTok{(}\OperatorTok{{-}}\FloatTok{1.0}\NormalTok{, n))}
  \KeywordTok{set.bounds}\NormalTok{(lp, }\DataTypeTok{lower =} \KeywordTok{rep}\NormalTok{(}\OperatorTok{{-}}\OtherTok{Inf}\NormalTok{, d), }\DataTypeTok{upper =} \KeywordTok{rep}\NormalTok{(}\OtherTok{Inf}\NormalTok{, d)) }\CommentTok{\# z unbounded}
  \KeywordTok{set.objfn}\NormalTok{(lp, p) }\CommentTok{\# objective function is p \%*\% z}
  \KeywordTok{solve}\NormalTok{(lp)}
  \KeywordTok{return}\NormalTok{(}\KeywordTok{get.objective}\NormalTok{(lp))}
\NormalTok{\}}
\end{Highlighting}
\end{Shaded}

Similarly, this code uses \texttt{Rglpk}:

\begin{Shaded}
\begin{Highlighting}[]
\KeywordTok{library}\NormalTok{(Rglpk)}
\NormalTok{LPmod2 \textless{}{-}}\StringTok{ }\ControlFlowTok{function}\NormalTok{(M, p) \{ }\CommentTok{\# it is assumed the centroid is the zero vector}
\NormalTok{  n \textless{}{-}}\StringTok{ }\KeywordTok{nrow}\NormalTok{(M)}
\NormalTok{  d \textless{}{-}}\StringTok{ }\KeywordTok{ncol}\NormalTok{(M)}
\NormalTok{  bounds \textless{}{-}}\StringTok{ }\KeywordTok{list}\NormalTok{(}
    \DataTypeTok{lower =} \KeywordTok{list}\NormalTok{(}\DataTypeTok{ind =} \KeywordTok{seq}\NormalTok{(d), }\DataTypeTok{val =} \KeywordTok{rep}\NormalTok{(}\OperatorTok{{-}}\OtherTok{Inf}\NormalTok{, d)),}
    \DataTypeTok{upper =} \KeywordTok{list}\NormalTok{(}\DataTypeTok{ind =} \KeywordTok{seq}\NormalTok{(d), }\DataTypeTok{val =} \KeywordTok{rep}\NormalTok{(}\OtherTok{Inf}\NormalTok{, d))}
\NormalTok{  )}
\NormalTok{  ans \textless{}{-}}\StringTok{ }\KeywordTok{Rglpk\_solve\_LP}\NormalTok{(}
    \DataTypeTok{obj =}\NormalTok{ p, }\DataTypeTok{mat =}\NormalTok{ M, }\DataTypeTok{dir =} \KeywordTok{rep}\NormalTok{(}\StringTok{"\textgreater{}="}\NormalTok{, n),}
    \DataTypeTok{rhs =} \KeywordTok{rep}\NormalTok{(}\OperatorTok{{-}}\FloatTok{1.0}\NormalTok{, n), }\DataTypeTok{bounds =}\NormalTok{ bounds}
\NormalTok{  )}
  \KeywordTok{return}\NormalTok{(ans}\OperatorTok{$}\NormalTok{optimum)}
\NormalTok{\}}
\end{Highlighting}
\end{Shaded}

A test case simulates \(n=100{,}000\) points uniformly distributed in
the \(d=20\)-dimensional unit cube, then verifies that the point
\(\pV=(1, \dotsc, 1)\T\) is not inside the convex hull by finding the
smallest \(\gamma>0\) such that \(\gamma\pV\) is on the boundary:

\begin{Shaded}
\begin{Highlighting}[]
\KeywordTok{set.seed}\NormalTok{(}\DecValTok{123}\NormalTok{)}
\NormalTok{M \textless{}{-}}\StringTok{ }\KeywordTok{matrix}\NormalTok{(}\KeywordTok{runif}\NormalTok{(}\FloatTok{1e5} \OperatorTok{*}\StringTok{ }\DecValTok{20}\NormalTok{), }\DataTypeTok{ncol =} \DecValTok{20}\NormalTok{)}
\NormalTok{Mbar \textless{}{-}}\StringTok{ }\KeywordTok{colMeans}\NormalTok{(M)}
\NormalTok{M \textless{}{-}}\StringTok{ }\KeywordTok{sweep}\NormalTok{(M, }\DecValTok{2}\NormalTok{, Mbar) }\CommentTok{\# Translate M so sample mean is at origin}
\NormalTok{p \textless{}{-}}\StringTok{ }\KeywordTok{rep.int}\NormalTok{(}\DecValTok{1}\NormalTok{, }\DecValTok{20}\NormalTok{) }\OperatorTok{{-}}\StringTok{ }\NormalTok{Mbar}
\NormalTok{lpstime \textless{}{-}}\StringTok{ }\KeywordTok{system.time}\NormalTok{(lpsstep \textless{}{-}}\StringTok{ }\DecValTok{{-}1} \OperatorTok{/}\StringTok{ }\KeywordTok{LPmod1}\NormalTok{(M, p))[}\DecValTok{1}\NormalTok{]}
\NormalTok{glpktime \textless{}{-}}\StringTok{ }\KeywordTok{system.time}\NormalTok{(glpkstep \textless{}{-}}\StringTok{ }\DecValTok{{-}1} \OperatorTok{/}\StringTok{ }\KeywordTok{LPmod2}\NormalTok{(M, p))[}\DecValTok{1}\NormalTok{]}
\KeywordTok{stopifnot}\NormalTok{(}\KeywordTok{isTRUE}\NormalTok{(}\KeywordTok{all.equal}\NormalTok{(lpsstep, glpkstep)))}
\end{Highlighting}
\end{Shaded}

Both implementations produce the same result: \(\gamma\approx0.4801\),
but \texttt{Rglpk} takes 10 seconds, whereas \texttt{lpSolveAPI} takes
19. However, \texttt{Rglpk} requires GLPK to be installed separately on
some platforms; thus, our suggestion for package developers is to check
for availability of GLPK on the system and use \texttt{Rglpk} if it is
installed.

\hypertarget{testing-multiple-points}{\section{\texorpdfstring{Testing multiple points
\label{sec:multiple}}{Testing multiple points }}\label{testing-multiple-points}}

As explained at the beginning of \secref{sec:LinearProgram}, the
approximated difference of log-likelihoods in \eqref{eq:LoglikApprox2}
requires that every \(\gV(\zM_i)\in\sS\) be contained in
\({\ConvexHull(\tS)}\) in order for the approximation to have a
maximizer. We might therefore consider a strategy of checking, prior to
using \eqref{eq:LoglikApprox2}, whether
\(\gV(\zM_i)\in{\ConvexHull(\tS)}\) for all \(i=1,\dotsc,s\) using the
single-test-point methods discussed earlier. This strategy has two
potential drawbacks: First, the computational burden might be quite high
if any of the dimension \(d\), the number of test points \(s\), or the
number of target points \(r\) is large. Second, we must decide what to
do in the case where one or more of the points in \(\sS\) is found to be
outside \({\ConvexHull(\tS)}\). This section addresses each of these
questions and then presents an illustrative example using a network
dataset with missing edges.

\hypertarget{reducing-computational-burden}{\subsection{\texorpdfstring{Reducing computational burden
\label{sec:reducing}}{Reducing computational burden }}\label{reducing-computational-burden}}

There are evidently many possible ways to approach the question of how
to efficiently decide which test points, if any, lie outside
\(\ConvexHull(\tS)\). Even if we only consider ideas for reducing the
size of the set \(\tS\) in such a way as to have little or no influence
on the answer, we might attempt to search for and eliminate target set
points either that lie entirely inside \(\ConvexHull(\tS)\), since
eliminating such points from \(\tS\) does not change
\(\ConvexHull(\tS)\) at all, or that lie close to other points in
\(\tS\), since eliminating such points does not change
\(\ConvexHull(\tS)\) very much. Here, we merely suggest a simplistic
version of the first approach, as a thorough exploration of methods for
reducing the size of \(\tS\) without altering \(\ConvexHull(\tS)\) is
well beyond the scope of this article.

Whether a point in \(\tS\) is interior to \(\ConvexHull(\tS)\) or
whether it lies on the boundary of \(\ConvexHull(\tS)\) has to do with
that point's so-called \emph{data depth}, a concept often attributed (in
the two-dimensional case) to \citet{tukey1975}. Because the deepest
points are the ones that can be eliminated without changing
\(\ConvexHull(\tS)\), the sizable literature on data depth \citep[see,
e.g., the discussion by][]{zuo2000} may be relevant to our problem. For
the \(d=2\)-dimensional case, several authors have developed efficient
methods for identifying exactly the points lying on the convex hull
boundary; this is a particular case of the so-called \emph{convex
layers} problem \citep{chazelle1985}.

Among the most simplistic ideas is to use Mahalanobis distance from the
centroid, defined for any point \(\xV\in\tS\) as \[
d(\xV) = \sqrt{ \xV\T \hat\Sigma\inv \xV},
\] where \(\hat\Sigma\) is the sample covariance matrix, as a measure of
data depth. For the \(|\tS|=100{,}000\)-point example of
\secref{sec:benchmark}, we also sample \(|\sS|=5\) corners of the unit
cube in \(\Reals^{20}\) and try eliminating the fraction \(f=i/10\) for
\(i=1, \dotsc, 10\) of the deepest points as measured by Mahalanobis
distance:

\begin{Shaded}
\begin{Highlighting}[]
\NormalTok{p \textless{}{-}}\StringTok{ }\KeywordTok{sweep}\NormalTok{(}\KeywordTok{matrix}\NormalTok{(}\KeywordTok{rbinom}\NormalTok{(}\DecValTok{5} \OperatorTok{*}\StringTok{ }\DecValTok{20}\NormalTok{, }\DecValTok{1}\NormalTok{, }\FloatTok{0.5}\NormalTok{), }\DataTypeTok{ncol =} \DecValTok{20}\NormalTok{), }\DecValTok{2}\NormalTok{, Mbar)}
\NormalTok{d \textless{}{-}}\StringTok{ }\KeywordTok{rowSums}\NormalTok{((M }\OperatorTok{\%*\%}\StringTok{ }\KeywordTok{solve}\NormalTok{(}\KeywordTok{cov}\NormalTok{(M))) }\OperatorTok{*}\StringTok{ }\NormalTok{M) }\CommentTok{\# Find squared Mahalanobis distances}
\NormalTok{M \textless{}{-}}\StringTok{ }\NormalTok{M[}\KeywordTok{order}\NormalTok{(d, }\DataTypeTok{decreasing =} \OtherTok{TRUE}\NormalTok{), ]}
\NormalTok{b \textless{}{-}}\StringTok{ }\KeywordTok{rep}\NormalTok{(}\OtherTok{Inf}\NormalTok{, }\DecValTok{10}\NormalTok{) }\CommentTok{\# Vector for storing LP minima}
\ControlFlowTok{for}\NormalTok{ (i }\ControlFlowTok{in} \DecValTok{10}\OperatorTok{:}\DecValTok{1}\NormalTok{) \{}
  \ControlFlowTok{for}\NormalTok{ (j }\ControlFlowTok{in} \DecValTok{1}\OperatorTok{:}\KeywordTok{nrow}\NormalTok{(p)) \{}
\NormalTok{    b[i] \textless{}{-}}\StringTok{ }\KeywordTok{min}\NormalTok{(b[i], }\KeywordTok{LPmod2}\NormalTok{(M[}\DecValTok{1}\OperatorTok{:}\NormalTok{(i }\OperatorTok{*}\StringTok{ }\KeywordTok{nrow}\NormalTok{(M) }\OperatorTok{/}\StringTok{ }\DecValTok{10}\NormalTok{), ], p[j, ]))}
\NormalTok{  \}}
\NormalTok{\}}
\end{Highlighting}
\end{Shaded}

\begin{figure}
\includegraphics[width=1\textwidth]{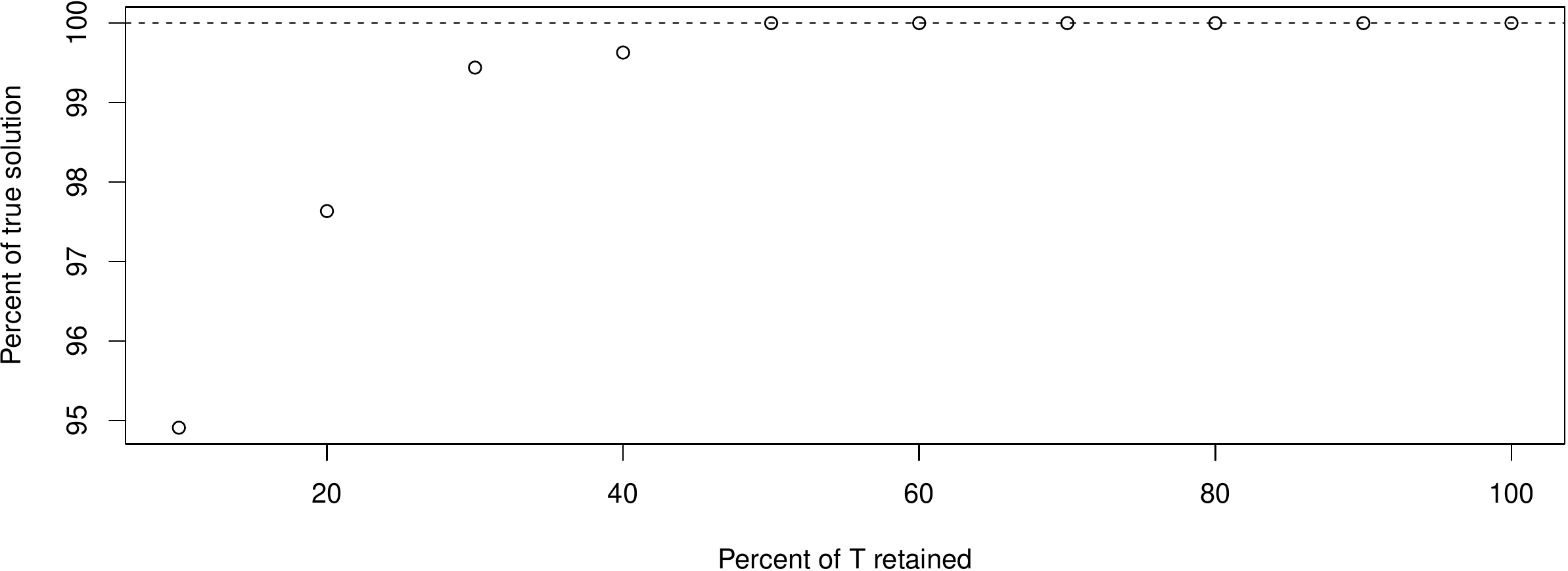} \caption{In this example where $d=20$ and $n=100{,}000$, half of the target set can be eliminated using a simplistic Mahalonobis distance-based algorithm without appreciably degrading the quality of the convex hull scaling factor.}\label{fig:plotTest}
\end{figure}

In this 20-dimensional problem, we see that we can effectively disregard
half the points in \(\tS\), using a simple Mahalanobis ordering, without
degrading the solution appreciably. Since computing effort scales
linearly with the number of points in \(\tS\), this seems like a useful
tool. We nonetheless recommend caution, as our experiments indicate that
the fraction of points that may be discarded by their Mahalanobis
distances varies considerably for different choices of \(d\). For
example, we find for \(d=50\) that each of the \(100{,}000\) points
sampled uniformly randomly in the unit cube lies on the boundary of the
convex hull.

\hypertarget{rescaling-observed-statistics}{\subsection{\texorpdfstring{Rescaling observed statistics
\label{sec:rescaling}}{Rescaling observed statistics }}\label{rescaling-observed-statistics}}

Recalling the case of a completely observed network, where \(\sS\)
consists of a single point \(\gV(\yM_\obs)\), the strategy used by
\citet{hummel2012} is to find a point, say \(\hat\xiV\), between \(\0V\)
and \(\gV(\yM_\obs)\) that lies inside \(\ConvexHull(\tS)\) yet close to
the boundary. By treating \(\hat\xiV\) as though it is \(\gV(\yM_\obs)\)
in constructing approximation \eqref{eq:LoglikApprox}, the approximation
will have a well-defined maximizer, and this maximizer is then used to
generate a new target set in the next iteration of the algorithm.

In analogous fashion, when \(\sS\) consists of multiple points and some
of them lie outside \(\ConvexHull(\tS)\), we propose to shrink each of
these points toward \(\0V\) using the same scaling factor, say
\(\gamma\), where \(\gamma\in(0,1]\) is chosen so that every element of
\(\sS\) lies within \(\ConvexHull(\tS)\) after the scaling is applied.
Since \lpref{LPretransformed} yields the optimal step length for each
point \(\pV\), we obtain this step length by iterating through the
points in \(\sS\) and selecting the least of the resulting step lengths.
This transformation also has the effect of shifting the sample mean of
the points in \(\sS\) to a point somewhere between \(\0V\) and the
sample mean of the untransformed points in \(\sS\).

\hypertarget{example}{\subsection{\texorpdfstring{Example
\label{sec:example}}{Example }}\label{example}}

This section illustrates the idea of \secref{sec:rescaling} using a
network representing a school friendship network based on a school
community in the rural western United States. The synthetic
\texttt{faux.mesa.high} network in the \texttt{ergm} package includes
205 students in grades 7 through 12. We may create a copy of this
network and then randomly select 10\% of the edge observations (whether
the edge is present or absent) to be categorized as missing:

\begin{Shaded}
\begin{Highlighting}[]
\KeywordTok{library}\NormalTok{(ergm)}
\KeywordTok{data}\NormalTok{(faux.mesa.high)}
\NormalTok{fmh \textless{}{-}}\StringTok{ }\NormalTok{faux.mesa.high}
\NormalTok{m \textless{}{-}}\StringTok{ }\KeywordTok{as.matrix}\NormalTok{(fmh)}
\NormalTok{el \textless{}{-}}\StringTok{ }\KeywordTok{cbind}\NormalTok{(}\KeywordTok{row}\NormalTok{(m)[}\KeywordTok{lower.tri}\NormalTok{(m)], }\KeywordTok{col}\NormalTok{(m)[}\KeywordTok{lower.tri}\NormalTok{(m)])}
\KeywordTok{set.seed}\NormalTok{(}\DecValTok{123}\NormalTok{)}
\NormalTok{s \textless{}{-}}\StringTok{ }\KeywordTok{sample}\NormalTok{(}\KeywordTok{nrow}\NormalTok{(el), }\KeywordTok{round}\NormalTok{(}\KeywordTok{nrow}\NormalTok{(el) }\OperatorTok{*}\StringTok{ }\FloatTok{0.1}\NormalTok{))}
\NormalTok{fmh[el[s, ]] \textless{}{-}}\StringTok{ }\OtherTok{NA}
\end{Highlighting}
\end{Shaded}

If we now define an ERGM using a few statistics related to those
originally used to create the \texttt{faux.mesa.high} network---details
are found via \texttt{help(faux.mesa.high)}---we begin with an estimator
that can be derived using straightforward logistic regression. We denote
this maximum pseudo-likelihood estimator or MPLE, details of which may
be found in Section 5.2 of \citet{hunter2008}, by \texttt{theta0} or
\(\thetaV_0\).

\begin{Shaded}
\begin{Highlighting}[]
\NormalTok{fmhFormula \textless{}{-}}\StringTok{ }\NormalTok{fmh }\OperatorTok{\textasciitilde{}}\StringTok{ }\NormalTok{edges }\OperatorTok{+}\StringTok{ }\KeywordTok{nodematch}\NormalTok{(}\StringTok{"Grade"}\NormalTok{) }\OperatorTok{+}
\StringTok{  }\KeywordTok{gwesp}\NormalTok{(}\KeywordTok{log}\NormalTok{(}\DecValTok{3} \OperatorTok{/}\StringTok{ }\DecValTok{2}\NormalTok{), }\DataTypeTok{fixed =} \OtherTok{TRUE}\NormalTok{)}
\NormalTok{theta0 \textless{}{-}}\StringTok{ }\KeywordTok{coef}\NormalTok{(}\KeywordTok{ergm}\NormalTok{(fmhFormula, }\DataTypeTok{estimate =} \StringTok{"MPLE"}\NormalTok{))}
\end{Highlighting}
\end{Shaded}

To construct approximation \eqref{eq:LoglikApprox2}, we need two samples
of random networks, \(\yM_{1},\dotsc,\yM_{r}\) and
\(\zM_{1},\dotsc,\zM_{s}\) from \(\SampleSpace\) and
\(\SampleSpace(\yM_{\obs})\), respectively. For the latter, we employ
the \texttt{constraints} capability of the \texttt{ergm} package.

\begin{Shaded}
\begin{Highlighting}[]
\NormalTok{gY \textless{}{-}}\StringTok{ }\KeywordTok{simulate}\NormalTok{(fmhFormula,}
  \DataTypeTok{coef =}\NormalTok{ theta0, }\DataTypeTok{nsim =} \DecValTok{500}\NormalTok{, }\DataTypeTok{output =} \StringTok{"stats"}\NormalTok{,}
  \DataTypeTok{control =} \KeywordTok{snctrl}\NormalTok{(}\DataTypeTok{MCMC.interval =} \FloatTok{1e4}\NormalTok{)}
\NormalTok{)}
\NormalTok{gZ \textless{}{-}}\StringTok{ }\KeywordTok{simulate}\NormalTok{(fmhFormula,}
  \DataTypeTok{coef =}\NormalTok{ theta0, }\DataTypeTok{nsim =} \DecValTok{100}\NormalTok{, }\DataTypeTok{output =} \StringTok{"stats"}\NormalTok{,}
  \DataTypeTok{constraints =} \OperatorTok{\textasciitilde{}}\NormalTok{observed,}
  \DataTypeTok{control =} \KeywordTok{snctrl}\NormalTok{(}\DataTypeTok{MCMC.interval =} \FloatTok{1e4}\NormalTok{)}
\NormalTok{)}
\end{Highlighting}
\end{Shaded}

We can use the code developed earlier in this article to show that the
\(\gV(\zM_j)\) statistics are not interior to the convex hull of the
\(\gV(\yM_i)\) statistics:

\begin{Shaded}
\begin{Highlighting}[]
\NormalTok{centroid \textless{}{-}}\StringTok{ }\KeywordTok{colMeans}\NormalTok{(gY)}
\NormalTok{gY \textless{}{-}}\StringTok{ }\KeywordTok{sweep}\NormalTok{(gY, }\DecValTok{2}\NormalTok{, centroid) }\CommentTok{\# Translate the g(Y) statistics}
\NormalTok{gZ \textless{}{-}}\StringTok{ }\KeywordTok{sweep}\NormalTok{(gZ, }\DecValTok{2}\NormalTok{, centroid) }\CommentTok{\# Translate the g(Y) statistics}
\NormalTok{scale \textless{}{-}}\StringTok{ }\OtherTok{Inf}
\ControlFlowTok{for}\NormalTok{ (j }\ControlFlowTok{in} \DecValTok{1}\OperatorTok{:}\KeywordTok{nrow}\NormalTok{(gZ)) \{}
\NormalTok{  scale \textless{}{-}}\StringTok{ }\KeywordTok{min}\NormalTok{(scale, }\DecValTok{{-}1} \OperatorTok{/}\StringTok{ }\KeywordTok{LPmod2}\NormalTok{(gY, gZ[j, ]))}
\NormalTok{\}}
\end{Highlighting}
\end{Shaded}

The code above finds that 0.733 is the largest scaling factor that, when
multiplied by each \(\gV(\zM_j)\) vector, ensures that the result is on
or inside the boundary of the convex hull of the target points
\(\gV(\yM_1),\dotsc, \gV(\yM_{100})\). Since this scaling factor is less
than 1, at least one element in the test set lies outside
\(\ConvexHull(\tM)\) and so the approximation in
\eqref{eq:LoglikApprox2} has no maximizer. \figref{fig:PairsPlot}
depicts the target points, the test points, and the test points after
scaling by 0.733.

\begin{figure}
\includegraphics[width=1\textwidth]{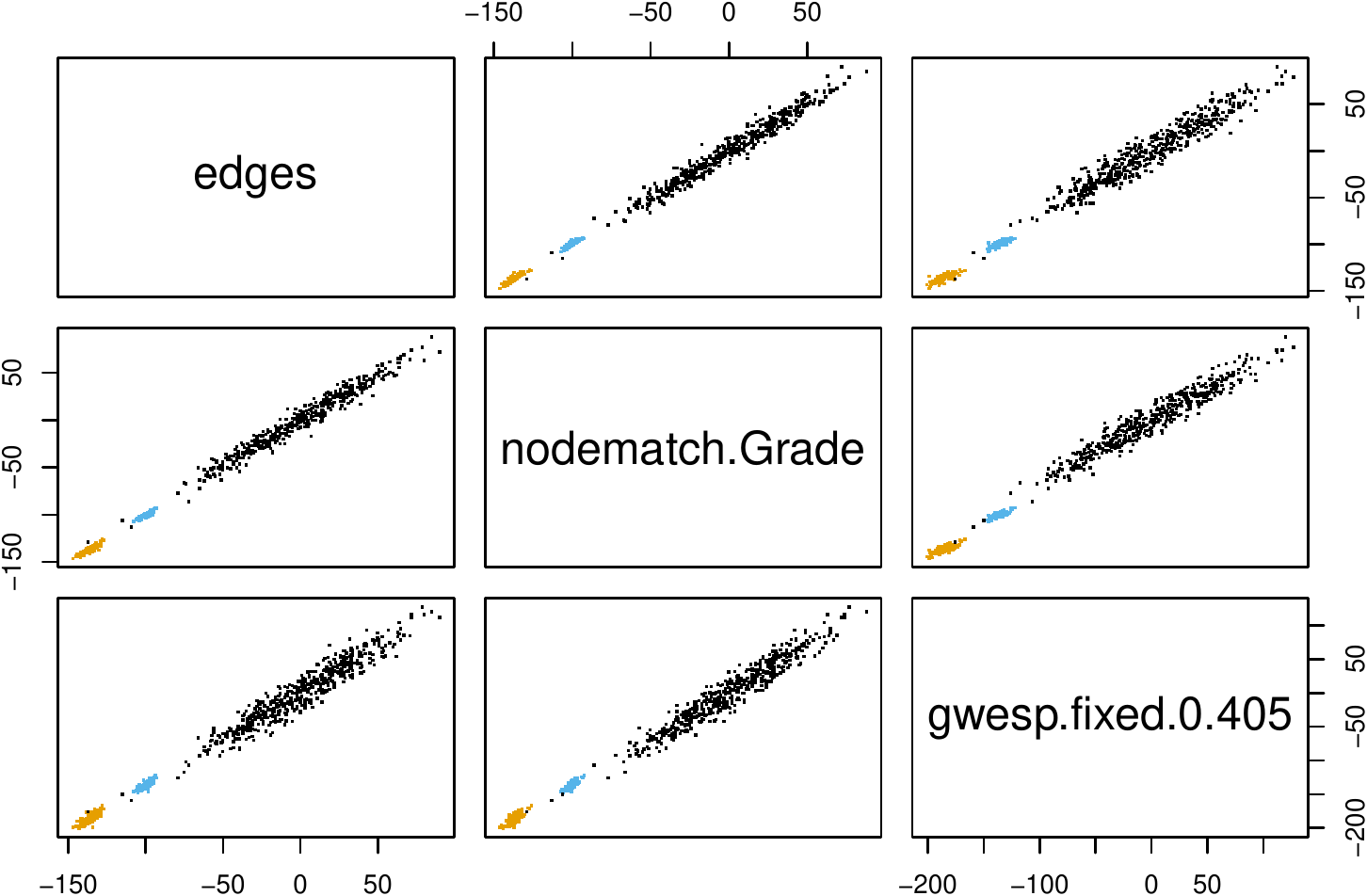} \caption{Pairwise scatterplots of three-dimensional network statistics generated using the MPLE, with statistics on the full sample space of networks as black dots and those on the sample space constrained to coincide with the observed network as orange dots. The constrained points are rescaled as blue points so that none lies outside the convex hull of the black points. In each plot, there are 500 black, 100 orange, and 100 blue points.}\label{fig:PairsPlot}
\end{figure}

With our samples in place, we may now construct approximation
\eqref{eq:LoglikApprox2}. Here, we negate the function so that our
objective is to minimize it:

\begin{Shaded}
\begin{Highlighting}[]
\NormalTok{l \textless{}{-}}\StringTok{ }\ControlFlowTok{function}\NormalTok{(ThetaMinusTheta0, gY, gZ, scale) \{}
  \OperatorTok{{-}}\KeywordTok{log}\NormalTok{(}\KeywordTok{mean}\NormalTok{(}\KeywordTok{exp}\NormalTok{(scale }\OperatorTok{*}\StringTok{ }\NormalTok{gZ }\OperatorTok{\%*\%}\StringTok{ }\NormalTok{ThetaMinusTheta0))) }\OperatorTok{+}
\StringTok{    }\KeywordTok{log}\NormalTok{(}\KeywordTok{mean}\NormalTok{(}\KeywordTok{exp}\NormalTok{(gY }\OperatorTok{\%*\%}\StringTok{ }\NormalTok{ThetaMinusTheta0)))}
\NormalTok{\}}
\end{Highlighting}
\end{Shaded}

Finally, we employ the \texttt{optim} function in \texttt{R} \citep{R}
to minimize the objective function. We use a scale value of 90\%~ of the
value that places one of the test values exactly on the boundary of the
convex hull so that all of the scaled test points are interior to the
convex hull. We may repeat the whole process iteratively until the
entire set of test points is well within the convex hull boundary:

\begin{Shaded}
\begin{Highlighting}[]
\NormalTok{multipliers \textless{}{-}}\StringTok{ }\NormalTok{scale}
\NormalTok{NewTheta \textless{}{-}}\StringTok{ }\NormalTok{theta0}
\ControlFlowTok{while}\NormalTok{ (scale }\OperatorTok{\textless{}}\StringTok{ }\FloatTok{1.11}\NormalTok{) \{ }\CommentTok{\# We want 0.9 * scale \textgreater{} 1 when finished}
\NormalTok{  NewTheta \textless{}{-}}\StringTok{ }\NormalTok{NewTheta }\OperatorTok{+}\StringTok{ }\KeywordTok{optim}\NormalTok{(}\DecValTok{0} \OperatorTok{*}\StringTok{ }\NormalTok{NewTheta, l,}
    \DataTypeTok{gY =}\NormalTok{ gY, }\DataTypeTok{gZ =}\NormalTok{ gZ,}
    \DataTypeTok{scale =} \FloatTok{0.9} \OperatorTok{*}\StringTok{ }\NormalTok{scale}
\NormalTok{  )}\OperatorTok{$}\NormalTok{par}
\NormalTok{  theta0 \textless{}{-}}\StringTok{ }\KeywordTok{rbind}\NormalTok{(theta0, NewTheta)}
\NormalTok{  gY \textless{}{-}}\StringTok{ }\KeywordTok{simulate}\NormalTok{(fmhFormula,}
    \DataTypeTok{coef =}\NormalTok{ NewTheta, }\DataTypeTok{nsim =} \DecValTok{500}\NormalTok{, }\DataTypeTok{output =} \StringTok{"stats"}\NormalTok{,}
    \DataTypeTok{control =} \KeywordTok{snctrl}\NormalTok{(}\DataTypeTok{MCMC.interval =} \FloatTok{1e4}\NormalTok{)}
\NormalTok{  )}
\NormalTok{  gZ \textless{}{-}}\StringTok{ }\KeywordTok{simulate}\NormalTok{(fmhFormula,}
    \DataTypeTok{coef =}\NormalTok{ NewTheta, }\DataTypeTok{nsim =} \DecValTok{100}\NormalTok{, }\DataTypeTok{output =} \StringTok{"stats"}\NormalTok{,}
    \DataTypeTok{constraints =} \OperatorTok{\textasciitilde{}}\NormalTok{observed,}
    \DataTypeTok{control =} \KeywordTok{snctrl}\NormalTok{(}\DataTypeTok{MCMC.interval =} \FloatTok{1e4}\NormalTok{)}
\NormalTok{  )}
\NormalTok{  centroid \textless{}{-}}\StringTok{ }\KeywordTok{colMeans}\NormalTok{(gY)}
\NormalTok{  gY \textless{}{-}}\StringTok{ }\KeywordTok{sweep}\NormalTok{(gY, }\DecValTok{2}\NormalTok{, centroid) }\CommentTok{\# Translate the g(Y) statistics}
\NormalTok{  gZ \textless{}{-}}\StringTok{ }\KeywordTok{sweep}\NormalTok{(gZ, }\DecValTok{2}\NormalTok{, centroid) }\CommentTok{\# Translate the g(Y) statistics}
\NormalTok{  scale \textless{}{-}}\StringTok{ }\OtherTok{Inf}
  \ControlFlowTok{for}\NormalTok{ (j }\ControlFlowTok{in} \DecValTok{1}\OperatorTok{:}\KeywordTok{nrow}\NormalTok{(gZ)) \{}
\NormalTok{    scale \textless{}{-}}\StringTok{ }\KeywordTok{min}\NormalTok{(scale, }\DecValTok{{-}1} \OperatorTok{/}\StringTok{ }\KeywordTok{LPmod2}\NormalTok{(gY, gZ[j, ]))}
\NormalTok{  \}}
\NormalTok{  multipliers \textless{}{-}}\StringTok{ }\KeywordTok{c}\NormalTok{(multipliers, scale)}
\NormalTok{\}}
\end{Highlighting}
\end{Shaded}

Table \ref{tab:iterationTable} gives successive values of \(\thetaV_0\)
that are determined as maximizers of \eqref{eq:LoglikApprox2} after the
test set points \(\gV(\zM_1), \dotsc, \gV(\zM_s)\) are rescaled to lie
within the convex hull of \(\gV(\yM_1), \dotsc, \gV(\yM_r)\). The
maximum pseudo-likelihood estimate (MPLE) in the first row of the table
is obtained using logistic regression and is often used as an initial
approximation to the maximum likelihood estimator when employing MCMC
MLE \citep{hunter2008}. However, the MPLE fails to take the missing
network observations into account and, as demonstrated by
\citet{hummel2012}, it is not the case that the MPLE generates sample
network statistics near the observed statistics.

\begin{table}

\caption{\label{tab:iterationTable}Values of \texttt{theta0} and their associated test point mulipliers, starting with the maximum pseudo-likelihood estimator.}
\centering
\begin{tabular}[t]{ccccc}
\toprule
Iteration & \texttt{edges} & \texttt{nodematch.Grade} & \texttt{gwesp.fixed.0.405} & Multiplier\\
\midrule
0 & -6.302 & 2.264 & 1.249 & 0.733\\
1 & -6.236 & 2.098 & 1.276 & 1.094\\
2 & -6.223 & 1.981 & 1.302 & 2.043\\
\bottomrule
\end{tabular}
\end{table}

Figure \ref{fig:PairsPlot2} shows that the final value of \(\thetaV_0\)
produces samples such that the whole test set lies on the interior of
the target set, which allows \eqref{eq:LoglikApprox2} to be maximized to
produce an approximate MLE. As recommended by \citet{hummel2012}, we
might use moderate-sized samples to obtain a viable \(\thetaV_0\) value
using the idea here, then use much larger samples once \(\thetaV_0\) has
been found in order to improve the accuracy of Approximation
\eqref{eq:LoglikApprox2}.

\begin{figure}
\includegraphics[width=1\textwidth]{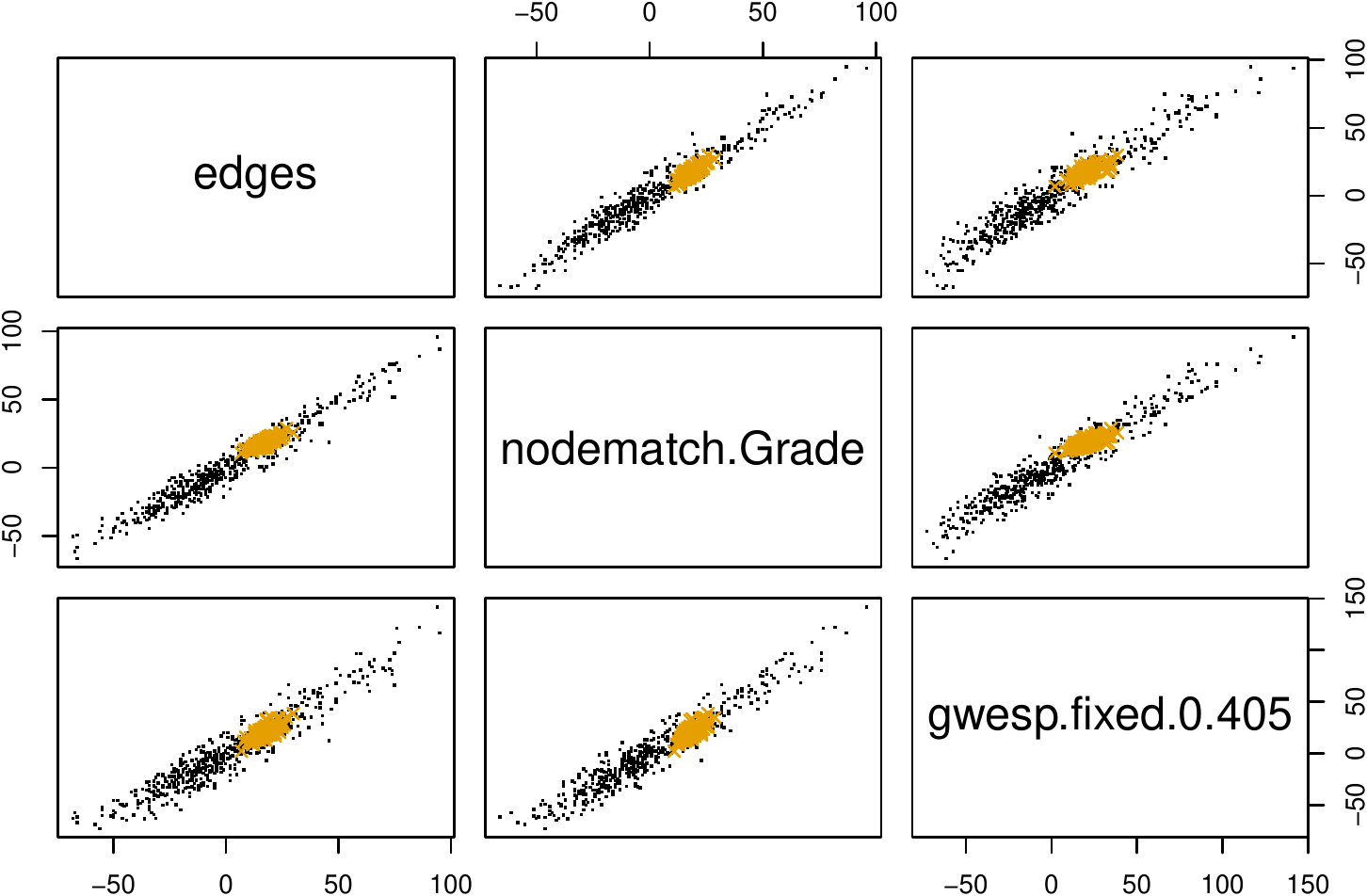} \caption{At the final iteration, the test set (orange) is entirely within the convex hull of the target set (black).}\label{fig:PairsPlot2}
\end{figure}

\hypertarget{discussion}{\section{\texorpdfstring{Discussion
\label{sec:discussion}}{Discussion }}\label{discussion}}

This article discusses the problem of determining whether a given point,
or set of points, lies within the convex hull of another set of points
in \(d\) dimensions. While this problem, along with its solution via
linear programming, is known, we are not aware of any work that
discusses it in the context of a statistical problem such as one
discussed here, namely, the maximization of an approximation to the
loglikelihood function for an intractable exponential-family model.

Here, we provide multiple improvements on the simplistic implementation
of the linear programming solution to the yes-or-no question involving a
single test point that exists in the \texttt{ergm} package
\citep{hunter2008} as a means of implementing the idea of
\citet{hummel2012}: First, we eliminate the need for the ``box
constraints'' of \texttt{ergm} and show how the dual linear program may
be derived from first principles. Second, we render the
``trial-and-error'' approach obsolete by showing how to find the exact
point of intersection between the convex hull boundary and the ray
originating at the origin and passing through the test point. Third, we
test the \texttt{lpSolveAPI} package that is currently used by
\texttt{ergm} against the \texttt{Rglpk} package, finding that the
latter appears to be far more efficient at solving the particular linear
programs we encounter. Fourth, we discuss the statistical case of
missing network observations, in which the test set may consist of
multiple points, establish an important necessary condition, and suggest
a method for handling this case.

In addition, we point out several ways in which this work might be
extended, particularly in the case of multiple test points. For one, the
question of how to streamline computations is wide open, particularly
since it is not necessary to find the exact maximum scaling factor that
maps each test point into the convex hull. For the purposes of
approximating a maximum likelihood estimator, we seek only an upper
bound on the acceptable scaling factors; indeed, in practice we want to
scale all test points so that they are inside the boundary. This means
that it might be possible to eliminate from the target set any points
that are sufficiently close to another target set point and that doing
so would not change the needed scaling factor too much.

We might also consider how to optimize the size of the sample chosen for
the test set in the first place. For instance, if the scaling factor
needed is considerably smaller than one, there might be an advantage in
sampling just a handful of points, possibly just a single point in order
to move the initial value of \(\thetaV_0\) closer to the true MLE, when
a larger sample of test points could be drawn.

One thing that is clear is that the extensions described here would be
much more difficult, if not impossible, to consider without the
improvements to the \texttt{ergm} package's convex-hull testing
procedure outlined in this article.

\hypertarget{acknowledgements}{\section*{Acknowledgements}\label{acknowledgements}}
\addcontentsline{toc}{section}{Acknowledgements}

Krivitsky was supported in part by US Army Research Office (Award
W911NF-21-1-0335 (79034-NS)).

\bibliographystyle{plainnat}
\bibliography{references.bib}

\end{document}